\documentclass[12pt,journal,onecolumn]{IEEEtran}
\usepackage{amsfonts,color,morefloats,pslatex}
\usepackage{amssymb,amsthm, amsmath,latexsym}

\newtheorem{theorem}{Theorem}
\newtheorem{lemma}[theorem]{Lemma}

\newtheorem{example}[theorem]{Example}

\newcommand{\tr}{{\mathrm{Tr}}}

\newcommand{\gf}{{\mathrm{GF}}}
\newcommand{\PG}{{\mathrm{PG}}}

\newcommand{\GAut}{{\mathrm{Aut}}}

\newcommand{\cP}{{\mathcal{P}}}
\newcommand{\cB}{{\mathcal{B}}}

\newcommand{\C}{{\mathcal{C}}}

\newcommand{\bc}{{\mathbf{c}}}

\newcommand{\bzero}{{\mathbf{0}}}

\newcommand{\bD}{{\mathbb{D}}}

\newcommand{\PGL}{{\mathrm{PGL}}}

\usepackage{blindtext}

\ifCLASSINFOpdf

\else

\fi

\begin{document}

\title{An infinite family of antiprimitive cyclic codes supporting Steiner systems $S(3,8, 7^m+1)$ \thanks{The research of C. Xiang was supported by the Basic Research Project of Science and Technology Plan of Guangzhou city of China  under grant number 202102020888 and the National Natural Science Foundation of China under grant number 12171162. The research of C. Tang was supported by National Natural Science Foundation
of China under grant number 11871058 and China West Normal University (14E013, CXTD2014-4 and the Meritocracy Research Funds).}
}

\author{
Can Xiang \thanks{C. Xiang is with the College of Mathematics and Informatics, South China Agricultural University, Guangzhou, Guangdong 510642, China (email:cxiangcxiang@hotmail.com). }% <-this % stops a space
, Chunming Tang \thanks{C. Tang  is with School of Mathematics  and Information,
China West Normal University, Nanchong, Sichuan 637002, China (email: tangchunmingmath@163.com).}
and Qi Liu \thanks{Q. Liu  is with School of Mathematics  and Information,
China West Normal University, Nanchong, Sichuan 637002, China (email: liuqijichushuxue@163.com).}
% <-this % stops a space
% <-this % stops a space
}

\maketitle

\begin{abstract}
Coding theory and combinatorial $t$-designs have close connections and interesting interplay. One of the major approaches to the construction of combinatorial t-designs is the employment of error-correcting codes.
As we all known, some $t$-designs have been constructed with this approach by using certain linear codes in recent years. However, only a few infinite families of cyclic codes holding an infinite family of $3$-designs are reported in the literature.
The objective of this paper is to study an infinite family of cyclic codes and determine their parameters. By the parameters of these codes and their dual, some infinite family of $3$-designs are presented and their parameters are also explicitly determined. In particular, the complements of the supports of the minimum
weight codewords in the studied cyclic code form a Steiner system. Furthermore, we show that the infinite family of cyclic codes admit $3$-transitive automorphism groups.
%by using the similar method in Liu et al.  (ArXiv:2109.09051, 2021).

\end{abstract}


\IEEEpeerreviewmaketitle

\section{Introduction}

Let $\gf(q)$ be a finite field with $q$ elements, where $q = p^m$ with $m$ being a positive integer and $p$ being an prime number.  An $[v,\, k,\,d]$ linear code $\C$ over $\gf(q)$ is a $k$-dimensional subspace of $\gf(q)^v$ with minimum (Hamming) distance $d$.
An $[v,\, k,\,d]$ linear code $\C$ is said to be {\em cyclic} if
$(c_0,c_1, \cdots, c_{v-1}) \in \C$ implies $(c_{v-1}, c_0, c_1, \cdots, c_{v-2}) \in \C$.

Let $\C$ be an $[v,k,d]$ cyclic code over $\gf(q)$. If $v=q^m-1$ (resp. $v=q^m+1$), the cyclic code $\C$ is called primitive (resp. antiprimitive).
If we identify a vector $(c_0,c_1, \cdots, c_{v-1}) \in \gf(q)^v$
with the following polynomial
$$\sum_{i=0}^{v-1} c_ix^i   \in \gf(q)[x]/(x^v-1),$$
then any cyclic code $\C$ of length $v$ over $\gf(q)$ is an ideal of the quotient ring $\gf(q)[x]/(x^v-1)$.
It is notice that the ring $\gf(q)[x]/(x^v-1)$ is a principal ideal ring. Thus, for any
cyclic code $\C$ of length $v$ over $\gf(q)$, there exists an unique monic divisor $g(x)$ of $x^v-1$ of the smallest degree such that $\C=\langle g(x) \rangle$ . This polynomial $g(x)$ is called the {\em generator polynomial,} and $h(x)=(x^v-1)/g(x)$ is called the {\em check} polynomial of $\C$.
It is obvious that $k=v-deg(g(x))$  and $\{g(x),x g(x), \cdots ,  x^{k-1} g(x) \}$ is a basis of $\C$.
It is well known that a cyclic code is a special linear code.
Although the error correcting capability of cyclic codes may not be as good as some other linear codes in general, cyclic codes have wide applications in storage and communication systems as they have efficient encoding and decoding algorithms \cite{Chien5,Forney12,Prange28}. Thus, cyclic codes have been attracted much attention in coding theory and a lot of progress has been made (see, for example, \cite{dinghell2013,zhouding2013,Liding2013,Ding2018,Zha2021,Yan01,Yan02}).

It is known that linear codes and $t$-designs are closely related. A $t$-design can be induced to a linear code (see, for example, \cite{Dingtv2020,Dingt20201}). Meanwhile, a linear code $\C$ may induce a $t$-design under certain conditions. As far as we know, a lot of $2$-designs and $3$-designs have been constructed from some special linear codes (see, for example, \cite{Dingbook18,Ding18dcc,DingLi16,Ding18jcd,TDX2019}). Recently, an infinity family of linear codes holding $4$-designs was settled by Tang and Ding in \cite{Tangding2020}. It remains open if there is an infinite family of linear codes holding $5$-designs. In fact, only a few infinite families of cyclic codes holding an infinite family of $3$-designs are reported in the literature. Motivated by this fact, we will consider a class of cyclic codes
\begin{eqnarray}\label{cm}
\C_{m}=\{(\tr(au^4+bu^3))_{u \in U_{q+1}}: a,b \in \gf(q^2)\}
\end{eqnarray}
over $\gf(q)$ and its dual, where $q=7^m$ with $m\geq 2$ being a integer, $\tr$ is the trace function from $\gf(q^2)$ to $\gf(q)$ and $U_{q+1}$ is the set of all $(q+1)$-th roots of unity in $\gf(q^2)$, and prove that these codes hold $3$-designs. Specifically, the cyclic code $\C_m$ and its dual $\C_m^ \perp $  admit $3$-transitive automorphism groups and the complement of the supports of the minimum weight codewords in $\C_{m}$  forms a steiner system $S(3,8, 7^m+1)$.

%In this paper, we will consider a family of cyclic codes and its dual, and prove that these codes support $3$-designs by the AM theorem.

The remainder of this paper is arranged as follows. Section \ref{sec-pre} introduces some notation and basics
of linear codes and combinatorial $t$-designs. Section \ref{sec-desbch} determines the parameters of the cyclic code $\C_m$ and its dual, and induces some infinite families of $3$-designs.
%Section \ref{sec-genaralbch} briefly summarizes the research progress of a class of BCH codes and their related codes.
Section  \ref{sec-summary} concludes this paper.

\section{Preliminaries}\label{sec-pre}

%Throughout this paper,
%It is well known that cyclic codes are a subclass of linear codes. Therefore, cyclic codes have all properties of linear codes.
As a special linear code, cyclic codes have all properties of linear codes.
In order to study cyclic codes in this paper, we need briefly introduce some known results on  linear codes and combinatorial t-designs in this section, which will be used later. For convenience, we begin this section
by fixing the following notations unless otherwise stated in this paper.

\begin{itemize}
  \item $p$ is a prime and  $q=p^m$ with $m$ being a positive integer.
  \item $\gf(q)$ is a finite field with $q$ elements and $\gf(q)^{*}=\gf(q)\setminus \{0\}$.
  \item $\tr$ is the trace function from $\gf(q^2)$ to $\gf(q)$.
  \item $U_{q+1}$ is the set of all $(q+1)$-th roots of unity in $\gf(q^2)$.
  \item $\binom{S}{k}$ is defined as the set consisting of all $k$-subsets of the set $S$ if $S$ is a set, and the binomial coefficient otherwise.
  \item $\PGL(2,q)$ is defined as the group
 of invertible $2\times 2$ matrices with entries in $\gf(q)$,
 modulo the scalar matrices
 $\begin{bmatrix}
a & 0\\
0 &  a
\end{bmatrix}$, where $a\in \gf(q)^*$.
\end{itemize}

\subsection{Weight enumerators of linear codes}

%It is well known that a cyclic code is a special linear code.
%It is well known that cyclic codes are a subclass of linear codes.
Let $\C$ be a $[v,\, k,\,d]$ linear code over $\gf(q)$.
Let $A_i$ denote the number of codewords with Hamming weight $i$ in a code
$\C$ for all  $0 \leq i \leq v$. The weight enumerator of $\C$ is defined by
$$
1+A_1z+A_2z^2+ \cdots + A_v z^v.
$$
The sequence $(1,A_1,\ldots,A_v)$ is called the weight distribution of $\C$. A code $\C$ is said to be a $t$-weight code  if the number of nonzero
$A_i$ in the sequence $(A_1, A_2, \cdots, A_v)$ is equal to $t$. A code $\C$ is said to be optimal if its parameters meet certain bounds on linear codes. Denote the dual of $\C$ by $\C^\perp$ ,
the minimum distance of $\C^\perp$ by $d^\perp$  and the weight distribution of $\C^\perp$ by $(A_0^\perp, A_1^\perp, \cdots, A_{v}^\perp)$. In order to determine the weight enumerator of $\C$, we will need the \emph{Pless power moments} \cite{HP10}. The first four Pless power moments identities are given by
\begin{align}\label{eq:PPM}
 & \sum_{i=0}^\nu  A_i= q^k, \nonumber \\
 & \sum_{i=0}^\nu  i\cdot A_i= q^{k-1} (qv-v-A_1 ^\perp), \nonumber \\
 & \sum_{i=0}^\nu  i^2 \cdot A_i= q^{k-2} \left [(q-1)v(qv-v+1)-(2qv-q-2v+2)A_1^\perp +2 A_2^\perp \right], \nonumber \\
 & \sum_{i=0}^\nu  i^3 \cdot A_i= q^{k-3}((q-1)v (q^2 v^2 -2q v^2 + 3q v - q + v^2 - 3 v + 2)  \nonumber \\
 &  ~~~~~~~~~~~~~~ - (3q^2 v^2 -3q^2 v - 6q v^2 + 12q v + q^2 - 6q + 3v^2 - 9v + 6) A_1^\perp       \nonumber \\
 &  ~~~~~~~~~~~~~~ + 6(qv- q- v +2)A_2^\perp -6 A_3 ^\perp              ).
 \end{align}

%Let $n$ be a positive integer such that $\gcd(n, q)=1$ and $m=\ord_{n}(q)$ be the order of $q$ modulo $n$, and let $\alpha$ denote a generator of the group $\gf(q^m)^*$. Define $\beta=\alpha^{(q^m-1)/n}$.
%Then $\beta$ is an $n$-th primitive root of unity in $\gf(q^m)$. The minimal polynomial $\m_{\beta^s}(x)$
%of $\beta^s$ over $\gf(q)$ is defined to be the monic polynomial of the least degree over the finite field $\gf(q)$ with
%$\beta^s$ as a root.

\subsection{Automorphism groups of linear codes}

Let $\C$ be a $[v,\, k,\,d]$ linear code over $\gf(q)$. We denote the set of coordinate positions of codewords of $\C$ by $\cP$. Then every codeword $\bc$ of $\C$ can be written as $\bc=(c_x)_{x\in \cP}$.
The set of coordinate permutations $g$ that map a code $\C$ to itself forms a group, i.e.,
\begin{eqnarray*}
\{g| ~g(c_x)_{x\in \cP} = (c_{g^{-1}(x)})_{x \in \cP} \in \C  \text{ for all }  (c_x)_{x\in \cP} \in \C \}
\end{eqnarray*}
which called the permutation automorphism group  of $\C$ and denoted by $ \mathrm{PAut}(\C)$ . We denote the symmetric group on the set $\cP$ by $\mathrm{Sym}(\cP)$.
It is clear that $ \mathrm{PAut}(\C)$ is the subgroup of
$\mathrm{Sym}(\cP)$ which keeps its invariance of the code $C$. Define a subgroup of $(\gf(q)^*)^v \rtimes  \mathrm{Sym}(\cP)$ as follows:
\begin{eqnarray} \label{eq-maut}
\{\left((a_x)_{x\in \cP};  g\right)|~\left((a_x)_{x\in \cP};  g\right) (c_x)_{x\in \cP}  = (a_{x} c_{g^{-1}(x)})_{x \in \cP} \in \C  \text{ for all }  (c_x)_{x\in \cP} \in \C \}
\end{eqnarray}
where $\left((a_x)_{x\in \cP};  g\right)$ is a map which maps the code $\C$ to itself. This subgroup is  called the monomial automorphism group of $\C$ and denoted by $\mathrm{MAut}(\C)$.
Let $\mathrm{Gal}(\gf(q))$ be the Galois group of $\gf(q)$ over its prime field. Then the automorphism group $\mathrm{Aut}(\C)$ of $\C$ is the subgroup of $(\gf(q)^*)^v \rtimes \left(  \mathrm{Sym}(\cP) \times \mathrm{Gal}(\gf(q))\right)$ as follows:
\begin{eqnarray*}
\{\left((a_x)_{x\in \cP};  g,  \gamma \right) | :~\left((a_x)_{x\in \cP};  g,  \gamma \right) (c_x)_{x\in \cP}  = (a_{x} \gamma(c_{g^{-1}(x)}))_{x \in \cP} \in \C  \text{ for all }  (c_x)_{x\in \cP} \in \C \}
\end{eqnarray*}
where $\left((a_x)_{x\in \cP};  g,  \gamma \right)$ is a map which maps the code $\C$  to itself. It is notice that $ \mathrm{PAut}(\C)\subseteq  \mathrm{MAut}(\C)\subseteq \mathrm{Aut}(\C)$. When $q$ is a prime,
$ \mathrm{MAut}(\C)=\mathrm{Aut}(\C)$. Specifically, $ \mathrm{PAut}(\C)=  \mathrm{MAut}(\C)=\mathrm{Aut}(\C)$ if $\C$ is binary code.

We say that $\GAut(\C)$ is  \textit{$t$-homogeneous\index{$t$-homogeneous}}
(resp. \textit{$t$-transitive\index{$t$-transitve}}) if for every
pair of $t$-element sets of coordinates (resp. $t$-element ordered sets of coordinates),  there is an element $\left((a_x)_{x\in \cP};  g, \gamma\right) \in \GAut(\C)$ such that its permutation part $g$ sends the first set to the second set.

\subsection{Combinatorial t-designs and some related results}

Let $k$, $t$ and $v$ be positive integers with $1 \leq t \leq k \leq  v$. Let $\cP$ be a set with $v$ elements and $\cB$ be a set of some $k$-subsets of
$\cP$. $\cB$  is called the point set and  $\cP$ is called the block set in general. The incidence structure
$\bD = (\cP, \cB)$ is called a $t$-$(v, k, \lambda)$ {\em design\index{design}} (or {\em $t$-design\index{$t$-design}}) if every $t$-subset of $\cP$ is contained in exactly $\lambda$ blocks of
$\cB$.
Let $\binom{\cP}{k}$ denote the set consisting of all $k$-subsets of the point set $\cP$. Then the incidence structure $(\cP, \binom{\cP}{k})$ is a $k$-$(v, k, 1)$ design and is called a \emph{complete design}. The special incidence structure $(\cP, \emptyset)$ is called a $t$-$(v, k, 0)$ trivial design
for all $t$  and $k$ . A combinatorial $t$-design is said to be {\em simple\index{simple}} if its block set $\cB$ does not have
a repeated block. When $t \geq 2$ and $\lambda=1$, a $t$-$(v,k,\lambda)$ design is called  a
{\em Steiner system\index{Steiner system}}
and denoted by $S(t,k, v)$. The parameters of a combinatorial $t$-$(v, k, \lambda)$ design must satisfy the following equation:
\begin{eqnarray}\label{eq:bb}
b  =\lambda \frac{\binom{v}{t}}{\binom{k}{t}}
\end{eqnarray}
where $b$ is the cardinality of $\cB$.

%The interplay between linear codes and $t$-designs has attracted a lot of attention for both directions. It is well known that the supports of all codewords with a fixed weight in a code may hold a t-design.
Linear codes and $t$-designs are closely related.
A $t$-design $\mathbb  D=(\mathcal P, \mathcal B)$ can be used to construct a linear code over GF($q$) for
any $q$  as follows.
Let $\mathcal P=\{q_1, \dots, q_{\nu}\}$, $\cB=\{B_1, \dots, B_b\}$.
The incidence matrix $M_{\bD}:=[m_{ij}]$ of the design $\bD=(\mathcal{P}, \mathcal{B})$ is a $b \times v$ binary matrix whose entry $m_{ij}=1$ if the point $q_j$
is on the block  $B_i$ and $m_{ij}=0$ otherwise.  The incidence matrix $M_{\bD}$ can be viewed as a matrix over
$\gf(q)$ for any $q$.
The \emph{linear code} $\mathsf{C}_{q}(\mathbb D)$ over the prime field $\mathrm{GF}(q)$ of the design $\mathbb D$ is defined to be the linear subspace of $\gf(q)^v$ spanned by the row vectors of the incidence matrix $M_{\bD}$.
Linear codes $\mathsf{C}_{q}(\mathbb D)$ of designs $\mathbb D$ have been extensively investigated (see, for example, \cite{Dingt20201,Ding15,ton1,ton2}).

On the other hand, a linear code $\C$ may produce a $t$-design which is formed by supports of codewords of a fixed Hamming weight in $\C$.
Let $\mathcal P(\mathcal C)=\{0,1, 2, \dots, \nu-1\}$ be the set of the coordinates of codewords in $\mathcal C$, where $\nu$ is the length of the code $\mathcal C$.
For a codeword $\mathbf c =(c_0, c_1, \dots, c_{\nu-1})$ in $\mathcal C$, the \emph{support} of  $\mathbf c$
is defined by
\begin{align*}
\mathrm{Supp}(\mathbf c) = \{i: c_i \neq 0, i \in \mathcal P(\mathcal C)\}.
\end{align*}
Let $\mathcal B_{w}(\mathcal C)$ denote the set $\{\{   \mathrm{Supp}(\mathbf c): wt(\mathbf{c})=w
~\text{and}~\mathbf{c}\in \mathcal{C}\}\}$, where $\{\{\}\}$ is the multiset notation. For some special code $\mathcal C$,
the incidence structure $\left (\mathcal P(\mathcal C),  \mathcal B_{w}(\mathcal C) \right)$
could be a $t$-$(v,w,\lambda)$ design for some positive integer $t$ and $\lambda$.
If  $\left (\mathcal P(\mathcal C),  \mathcal B_{w}(\mathcal C) \right)$ is a $t$-design for all $w$ with $0\le w \le \nu$,
we say that the code $\mathcal C$ \emph{supports $t$-designs}. By definition, such design
$\left (\mathcal P(\mathcal C),  \mathcal B_{w}(\mathcal C) \right)$ could have some repeated
blocks, or could be simple, or may be trivial.
In this way, many $t$-designs have been constructed  from linear codes (see, for example, \cite{Dingbook18,Dingbook18,Ding18dcc,DingLi16,Ding18jcd,TDX2019,Tangding2020}). A major way to construct combinatorial $t$-designs with linear codes over finite fields is the use of linear codes with $t$-transitive or $t$-homogeneous automorphism groups (see \cite[Theorem 4.18]{Dingbook18}) and some combinatorial $t$-designs (see, for example, \cite{LiuDing2017}) were obtained by this way. Very recently, Liu et al.\cite{Liudingtang2021} obtained some $3$-transitive automorphism groups from a class of BCH codes and derived some combinatorial $3$-designs with this way.
Another major way to construct $t$-designs with linear codes is the use of the
Assmus-Mattson Theorem (AM Theorem for short) in \cite[Theorem 4.14]{Dingbook18} and the generalized version of the
AM Theorem in \cite{Tangit2019}, which was recently employed to construct a number of $t$-designs (see, for example, \cite{Dingbook18,du1,du2}).
The following theorem is a generalized version of the
AM Theorem, which was developed in \cite{Tangit2019} and will be needed in this paper.

\begin{theorem}\cite{Tangit2019}\label{thm-designGAMtheorem}
Let $\mathcal C$ be a linear code over the finite field $\mathrm{GF}(q)$ with length $\nu$ and minimum distance $d$.
Let $\mathcal C^{\perp}$ denote the dual of $\mathcal C$ with minimum distance $d^{\perp}$.
Let $s$ and $t$ be two positive integers such that $t< \min \{d, d^{\perp}\}$. Let $S$ be a $s$-subset
of the set $\{d, d+1, d+2, \ldots, \nu-t  \}$.
Suppose that
$\left ( \mathcal P(\mathcal C), \mathcal B_{\ell}(\mathcal C) \right )$ and $\left ( \mathcal P(\mathcal C^{\perp}), \mathcal B_{\ell^{\perp}}(\mathcal C^{\perp}) \right )$
are $t$-designs  for
$\ell    \in \{d, d+1, d+2, \ldots, \nu-t  \} \setminus S $ and $0\le \ell^{\perp} \le s+t-1$, respectively. Then
the incidence structures
 $\left ( \mathcal P(\mathcal C) , \mathcal B_k(\mathcal C) \right )$ and
  $\left ( \mathcal P(\mathcal C^{\perp}), \mathcal B_{k}(\mathcal C^{\perp}) \right )$ are
  $t$-designs for any
$t\le k  \le \nu$, and particularly,
\begin{itemize}
\item the incidence structure $\left ( \mathcal P(\mathcal C) , \mathcal B_k(\mathcal C) \right )$ is a simple $t$-design
      for all integers $k$ with $d \leq k \leq w$, where $w$ is defined to be the largest  integer
      such that $w \leq \nu$ and
      $$
      w-\left\lfloor \frac{w+q-2}{q-1} \right\rfloor <d;
      $$
\item  and the incidence structure $\left ( \mathcal P(\mathcal C^{\perp}), \mathcal B_{k}(\mathcal C^{\perp}) \right )$ is
       a simple $t$-design
      for all integers $k$ with $d \leq k \leq w^\perp$, where $w^\perp$ is defined to be the largest integer
      such that $w^\perp \leq \nu$ and
      $$
      w^\perp-\left\lfloor \frac{w^\perp+q-2}{q-1} \right\rfloor <d^\perp.
      $$
\end{itemize}
\end{theorem}

%It was demonstrated in \cite{Tangit2019} that the generalized AM theorem (i.e., Theorem \ref{thm-designGAMtheorem})
%can outperform the original AM Theorem. Specifically, it was shown in \cite{Tangit2019} that Theorem \ref{thm-designGAMtheorem} can be used to prove that three classes of of linear codes support $2$-designs or $3$-designs,
%but the original AM theorem cannot be used to do so.

\section{An infinite family of cyclic codes supporting $3$-designs }\label{sec-desbch}

In this section, our task is to establish the parameters of the cyclic code $\C_{m}$ defined by (\ref{cm}) and its dual, and prove that these codes hold $3$-designs and satisfy $3$-transitive automorphism groups. To this end, we shall
prove a few more auxiliary results before proving the main results (see Theorems \ref{thm-d4}, \ref{thm-cm} and \ref{thm:C-group}) of this paper.

\subsection{Some auxiliary results}

%We will need the following three lemmas.
In order to determine the minimum distance of the dual code $\C_m ^\perp$ of $\C_m$, we need the results in the following three lemmas.
\begin{lemma}\label{lem-5neq0}
Let symbols and notation be the same as before. Let $m\geq 2$ be a positive integer and  $q=7^m$. For any $\{x, y, z\} \in \binom{U_{q+1}}{3}$, we have the following results.
\begin{enumerate}
  \item $x+y-2z \neq 0$;
  \item $x+2y-3z \neq 0$;
  \item $x+3y-4z \neq 0$;
  \item $x+4y-5z \neq 0$;
  \item $x+5y-6z \neq 0$;.
\end{enumerate}
\end{lemma}

\begin{proof}
We only give the proof of the first conclusion. The proofs of the other four conclusions are similar to that of the first conclusion and thus omitted.

Assume that $x+y-2z =0$, then
$$
(x+y-2z)^q=\frac{1}{x}+\frac{1}{y}-\frac{2}{z}=\frac{1}{x}+\frac{1}{y}-\frac{4}{2z}=0.
$$
It follows from $x+y=2z $ that
$$
\frac{1}{x}+\frac{1}{y}-\frac{4}{x+y}=0,
$$
which means that $(x-y)^2=0$. This is contrary to our assumption that $x,y,z$ are pairwise distinct. Thus, $x+y-2z \neq 0$. This completes the proof.
\end{proof}

\begin{lemma}\label{lem-det3}
Let symbols and notation be the same as before. Let $q=7^m$ with $m\geq 2$ being a positive integer and  $\{x, y, z\} \in \binom{U_{q+1}}{3}$.
Define
\begin{eqnarray}\label{eq-M3}
\bar{M}(x,y,z)=
\left[
\begin{array}{lll}
1  & 1 & 1 \\
x  & y & z \\
x^7 & y^7 & z^7
\end{array}
\right].
\end{eqnarray}
Then
$$|\bar{M}(x,y,z)| =(x - y) (x - z) (y - z)(x+y-2z)(x+2y-3z)(x+3y-4z)(x+4y-5z)(x+5y-6z) \neq 0.$$

\end{lemma}

\begin{proof}
The conclusion follows from Lemma \ref{lem-5neq0}.
\end{proof}

\begin{lemma}\label{lem-unq5w}
Let $m\geq 2$ be a positive integer, $q=7^m$ and $(x, y, z, w) \in \gf(q^2)^4$. Define
\begin{eqnarray}\label{eq-M4}
M(x,y,z,w)=
\left[
\begin{array}{llll}
1  & 1 & 1 & 1\\
x  & y & z & w\\
x^7 & y^7 & z^7& w^7 \\
x^8 & y^8 & z^8& w^8
\end{array}
\right].
\end{eqnarray}
Then we have the following results.
\begin{enumerate}
  \item $|M(x,y,z,w)|= (x-y)(x-z)(x-w)(y-z)(y-w)(z-w) \cdot \prod_{i=1}^{5}(xy+zw+i(xz+wy)+(6-i)(xw+yz) )$.

  \item For any $\{x, y, z\} \in \binom{U_{q+1}}{3}$, there exists five pairwise distinct $w\in U_{q+1} \setminus \{x,y,z\}$ such that $|M(x,y,z,w)|=0$, i.e., $\prod_{i=1}^{5}(xy+zw+i(xz+wy)+(6-i)(xw+yz) )=0$.
\end{enumerate}
\end{lemma}

\begin{proof}
1) It is easy to prove the first conclusion and we omit its proof.

2) By definitions and Lemma \ref{lem-5neq0}, $z+iy+(6-i)x \neq 0$ for any $i\in \{1,2,3,4,5\}$.  Denote $$w_i=\frac{(i-6)yz-xy-ixz}{z+iy+(6-i)x}$$,
where $i\in \{1,2,3,4,5\}$. Note that
$$w_i^{q}= \frac{(i-6)y^q z^q-x^q y^q-i x^q z^q }{z^q +i y^q+(6-i)x^q}=\frac{((i-6)y^q z^q-x^q y^q-i x^q z^q)\cdot xyz }{(z^q +i y^q+(6-i)x^q))\cdot xyz}=1/w_i. $$
Thus, $w_i^{q+1}=1$. This means that $w_i \in U_{q+1}$ for any $i\in \{1,2,3,4,5\}$.

Since $\{x, y, z\} \in \binom{U_{q+1}}{3}$, from $|M(x,y,z,w)|=0$ and the first conclusion 1) we have

$$\prod_{i=1}^{5}(xy+zw+i(xz+wy)+(6-i)(xw+yz) )=0.$$
This means that $w=w_i$ with $i\in \{1,2,3,4,5\}$.

Next we will prove that $w_i \neq x, y, z$ for each $i\in \{1,2,3,4,5\}$.

Let $i\in \{1,2,3,4,5\}$. Suppose that $w_i =x$, then
$$\frac{(i-6)yz-xy-ixz}{z+iy+(6-i)x}=x,
$$
which yields
$$(i-6)x^2-(i+1)(y+z)x+(i-6)yz=0,$$
which is the same as
$$(i-6)(x-y)(x-z)=0.$$
This means that $x=y$ or $x=z$, which is contrary to our assumption that $x,y,z$ are pairwise distinct in $U_{q+1}$. Thus, $w_i \neq x$.  Due
to symmetry, $w_i \neq y$ and $w_i \neq z$. Therefore, $w_i \neq x, y, z$ for each $i\in \{1,2,3,4,5\}$.

We now prove that $w_i\neq w_j$ when $i\neq j$ and $i,j\in \{1,2,3,4,5\}$.

Let $i,j\in \{1,2,3,4,5\}$ and $i\neq j$. Suppose that $w_i =w_j$, then

$$\frac{(i-6)yz-xy-ixz}{z+iy+(6-i)x}= \frac{(j-6)yz-xy-jxz}{z+jy+(6-j)x},
$$
which yields
$$
(i-j)(x-y)(x-z)(y-z)=0.
$$
It then follows from $i\neq j$ that $x=y$ or $x=z$ or $y=z$, which is contrary to our assumption that $x,y,z$ are pairwise distinct in $U_{q+1}$. Thus, $w_1,w_2,w_3,w_4,w_5$ are pairwise distinct. This completes the proof.
\end{proof}

%In order to determine the parameters of the cyclic code $\C_{m}$
The following result plays an important role in calculating the weight distributions of the cyclic code $\C_m$, which is described in the next lemma.

\begin{lemma}\label{lem-zero}
Let symbols and notation be the same as before. Let $m\geq 2$ be a positive integer, $q=7^m$, $(a,b)\in \gf(q^2)^2 \setminus \{(0,0)\}$ and $f(u)=\tr(au^4+bu^3)$. Define
$$Zero(f)=\{u\in U_{q+1}: f(u)=0\}.$$
Then $\# Zero(f)\leq 8$.  In particular, $\# Zero(f)=8$ if $\# Zero(f) \geq 3$.
\end{lemma}

\begin{proof}
Recall that $\tr$ is the trace function from $\gf(q^2)$ to $\gf(q)$.
By definition, then
$$f(u)=\tr(au^4+bu^3)=\frac{1}{u^4}(au^8+bu^7+b^q u +a^q) $$
if $u\in U_{q+1}$.
Thus, $\# Zero(f)\leq 8$ and
$$Zero(f)=\{u\in U_{q+1}: au^8+bu^7+b^q u +a^q=0\}.$$

If $\# Zero(f) \geq 3$, we assume that $\{1,-1,u_0\}\subseteq Zero(f)$. Then $f(1)=f(-1)=f(u_0)=0$, which yields
\begin{eqnarray}\label{eq:u0}
\left\{
\begin{array}{l}
a^q=-a  \\
b^q=-b  \\
au_0^8+bu_0 ^7-b u_0 -a=0.
\end{array}
\right.
\end{eqnarray}
Denote $g(u)=au^8+bu^7+b^q u +a^q$. Assume that $a \neq 0$ and denote $c=b/a$. By the first two equations in (\ref{eq:u0}), we have
$c^q=(b/a)^q=c$ and
$$
g(u)=a(u^8+\frac{b}{a}u^7-\frac{b}{a}u-1)=a(u^8+cu^7-cu-1).
$$
Denote
$$h(u)=u^8+cu^7-cu-1.$$
Then it is not hard to verify that $u^q$ and $u^{-1}$ are also the roots of $h(u)=0$ if $u$ is a root of $h(u)=0$.
%Note that $u_0$ is the root of $h(u)=0$ and $u_0 \in U_{q+1}$.
Thus, $1$, $-1$, $u_0$ and $u_0^{-1}=u_0 ^q$ are also the roots of $h(u)=0$.
Further, from the third equation in (\ref{eq:u0}), we have
$$
c=\frac{u_0^8-1}{u_0^7-u_0}.
$$
Thus,
$$
h(u)=u^8+\frac{u_0^8-1}{u_0^7-u_0} \cdot u^7-\frac{u_0^8-1}{u_0^7-u_0} \cdot u-1.
$$
If $h(u)=0$, then
$$
(u_0 ^7-u^0)(u^8-1)-(u_0 ^8-1)(u^7-u)=0
$$
which is the same as
\begin{eqnarray*}
\begin{array}{ccc}
(u_0^2-1)(u^2-1)(u-u_0)(u_0 u-1)& \cdot (u_0 u+5u+2u_0-1) & ~\\
~                               & \cdot (u_0 u+4u+3u_0-1) & ~\\
~                               & \cdot (u_0 u+3u+4u_0-1) & ~\\
~                               & \cdot (u_0 u+2u+5u_0-1) & =0.
\end{array}
\end{eqnarray*}
This means that $h(u)=0$ has eight roots as follows:
\begin{eqnarray}\label{eq-8roots}
1,~-1,~u_0,~u_0 ^{-1},~u_1, ~u_2, ~ u_3, ~u_4,
\end{eqnarray}
where
\begin{eqnarray}\label{eq:u14}
\left\{
\begin{array}{l}
u_1=\frac{-2u_0+1}{u_0-2} \\
u_2=\frac{-3u_0+1}{u_0-3} \\
u_3=\frac{-4u_0+1}{u_0-4} \\
u_4=\frac{-5u_0+1}{u_0-5}.
\end{array}
\right.
\end{eqnarray}
Note that it is not hard to verify that $u_3=u_1^{-1}$, $u_4=u_2^{-1}$ and
$u_i ^q=1/u_i$ for any $i\in\{1,2,3,4\}$, which means that
$u_i^{q+1}=1$, i.e., $u_i\in U_{q+1}$ for any $i\in\{1,2,3,4\}$. Thus, the eight roots given by (\ref{eq-8roots}) are
in $U_{q+1}$. Moreover, these eight roots in (\ref{eq-8roots}) are pairwise distinct. Suppose that $u_1 =u_0$,
then we have $u_0 ^2=1$. This means that $u_0=1$ or $u_0=-1$, which is contrary to our assumption that $1,-1,u_0$ are pairwise distinct in $U_{q+1}$. Thus, $u_1 \neq u_0$.  By similar discussions, it is easily obtain that all elements in (\ref{eq-8roots}) are pairwise distinct. This completes the proof.
%The desired conclusions then follow.
\end{proof}

%\subsection{main}
\subsection{The parameters of cyclic codes}
In this subsection, we will determine the parameters of  the cyclic code $\C_{m}$  and its dual $\C_{m}^\perp$, and prove that these codes hold $3$-designs.

\begin{theorem}\label{thm-d4}
Let $q=7^m$ with $m \geq 2$ being a positive integer. Then the code $\C_{m}^\perp$ over $\gf(q)$
has the parameters $[q+1, q-3, 4]$. Furthermore, the minimum weight codewords in $\C_{m}^\perp$ support a $3$-$(q+1, 4, 5)$ design.
\end{theorem}

\begin{proof}
It follows from definitions that  the code $\C_{m}^\perp$ has length $q+1$.
Let $\alpha$ be a generator of the multiplicative group $\gf(q^2)^*$ and define $\beta=\alpha^{q-1}$. Then $\beta \in U_{q+1}$ is an $q+1$-th primitive root of unity
in the field $\gf(q^2)$. Let $g_i (x)$ denote the
minimal polynomial of $\beta^{i}$ over $\gf(q)$, where $i \in \{3,4\}$. Note that $g_i (x)$  has only the roots $\beta^{i}$ and
$\beta^{-i}$. We then deduce that $g_3 (x)$ and $g_4 (x)$ are pairwise distinct irreducible polynomials of degree
2. By definition, the generator polynomial of $\C_{m}^\perp$ is $g_3 (x) g_4 (x) $ with degree 4.
Thus, $\C_{m}^\perp$ has dimension $q+1-4=q-3$.

Let $U_{q+1}=\{x_1,x_2,x_3...,x_{q+1}\}$. Define
\begin{eqnarray}\label{eq:H}
H=\left[
\begin{array}{rrrrr}
x_1^{-4}  & x_2 ^{-4}  & x_3^{-4}  & \cdots & x_{q+1}^{-4}  \\ [2mm]
x_1^{-3}  & x_2 ^{-3}  & x_3^{-3}  & \cdots & x_{q+1}^{-3}  \\ [2mm]
x_1^{~3}    & x_2 ^{~3} & x_3 ^{~3} & \cdots & x_{q+1} ^{~3} \\ [2mm]
x_1^{~4}   & x_2 ^{~4}  & x_3^{~4}  & \cdots & x_{q+1}^{~4}
\end{array}
\right].
\end{eqnarray}
It is easily observed that
\begin{eqnarray}\label{eq:cmperp}
\C_{m}^\perp=\{\bc \in \gf(q)^{q+1}: \bc H^T=\bzero\}.
\end{eqnarray}
By Lemma \ref{eq-M3} and Equations (\ref{eq:H}) and (\ref{eq:cmperp}), we have the minimum distance $d$
of $\C_{m}^\perp $ is at least $4$. Next we will prove that $d=4$.

Let $\{x, y, z, w\} \in \binom{U_{q+1}}{4}$. Without the loss of generality,
we assume that
$$
x=x_{i_1}, \ y=x_{i_2}, \ z=x_{i_3},  \ w=x_{i_4},
$$
where $1 \leq i_1<i_2<i_3<i_4 \leq q+1$.
%such that (\ref{eqn-tcm161}) holds.
Since $d \geq 4$, the rank of the matrix
$M(x,y,z,w)$
equals $3$, where $M(x,y,z,w)$ was defined by (\ref{eq-M4}). Let $(u_{i_1}, u_{i_2}, u_{i_3}, u_{i_4})\in \gf(q)^4$ denote a nonzero solution of
\begin{eqnarray*}
\left[
\begin{array}{llll}
1  & 1 & 1 & 1\\
x  & y & z & w\\
x^7 & y^7 & z^7& w^7 \\
x^8 & y^8 & z^8& w^8
\end{array}
\right]
\left[
\begin{array}{llll}
u_{i_1} \\
u_{i_2} \\
u_{i_3} \\
u_{i_4}
\end{array}
\right]
= \bzero.
\end{eqnarray*}
Since the rank of the matrix $M(x,y,z,w)$ is $3$, all these $u_{i_j} \neq 0$.
Define a vector $\bc=(c_0, c_1, \ldots, c_n) \in \gf(q)^{n+1}$, where $c_{i_j}=u_{i_j}$ for $j \in \{1,2,3,4\}$
and $c_h =0$ for all $h \in \{0,1, \ldots, n\} \setminus \{i_1, i_2, i_3, i_4\}$. It is easily observed that $\bc$ is a
codeword with Hamming weight $4$ in $\C_{m}^\perp $. The set $\{a\bc: a \in \gf(q)^*\}$ consists of all such codewords
of Hamming weight $4$ with nonzero coordinates in $\{i_1, i_2, i_3, i_4\}$. Hence, the code $\C_{m}^\perp $ has minimum distance $d=4$. Meanwhile, every codeword of Hamming weight $4$ in $\C_{m}^\perp $ with
nonzero coordinates in $\{i_1, i_2, i_3, i_4\}$ must correspond to the set $\{x,y,z,w\}$. Further, from $|M(x,y,z,w)|=0$ and Lemma \ref{lem-unq5w}, it follows that every codeword of weight $4$ and its nonzero multiples in $\C_{m}^\perp $ correspond to five such set $\{x,y,z,w\}$ . We then deduce that the codewords of
weight $4$ in $\C_{m}^\perp $ support a $3$-$(q+1, 4, 5)$ design. Thus, the number of the codewords of weight $4$ in $\C_{m}^\perp $ is
$$
A_4^\perp=(q-1)\cdot  \frac{\binom{q+1}{3}}{\binom{4}{3}} \cdot 5 = \frac{5(q-1)^2 q (q+1)}{24}.
$$
This completes the proof.
\end{proof}

\begin{example}\label{exam-1}
Let $m=2$. Then the code $\C_{m}^\perp $ has the parameters $[50, 46, 4]$. The number of the codewords of weight $4$ in $\C_{m}^\perp $ is $A_4^\perp=1176000$.
The codewords of weight $4$ in $\C_{m}^\perp $ support a $3$-$(50,4,5)$ design.
\end{example}

It is now time to determine the parameters of the cyclic code $\C_{m}$, which is described in the following theorem.

\begin{theorem}\label{thm-cm}
Let $q=7^m$ with $m \geq 2$ being a integer. Then we have the following results.
\begin{enumerate}
  \item [(\uppercase\expandafter{\romannumeral1})]
The code $\C_{m}$ over $\gf(q)$
has the parameters $[q+1, 4, q-7]$
and the weight enumerator
\begin{eqnarray}\label{cmw}
1+ \frac{1}{336} (q-1)^2 q (q+1) z^{q-7} + \frac{1}{12}(q-1) q (1 + q) (7 + 5 q) z^{q-1}  +  \\
 ~~~\frac{1}{7}(q-1) (1 + q) (7 + ( q-1) q) z^{q} + \frac{7}{16} ( q-1)^2 q (1 + q)  z^{q+1}.            \nonumber
\end{eqnarray}
  \item [(\uppercase\expandafter{\romannumeral2})] The code $\C_{m}$ and its dual $\C_{m}^\perp$ support $3$-designs. Furthermore, the codewords of weight $q-7$ in $\C_{m}$  hold a $3$-$(q+1, q-7, \lambda)$ design, where
  $$
  \lambda=\frac{(q-7)(q-8)(q-9)}{336}.
  $$
The complement of this design is a $3$-$(q+1,8,1)$ design, i.e., Steiner systems $S(3,8,q +1)$.
\end{enumerate}
\end{theorem}

\begin{proof}
(\uppercase\expandafter{\romannumeral1}) By definition, it is clear that  the code $\C_{m}$  has length $q+1$. By Theorem \ref{thm-d4}, the dual code $\C_{m}^\perp$ of $\C_{m} $ has dimension $q-3$. Thus, the dimension of the code $\C_{m}$ is 4.

Further, by definitions and Lemma \ref{lem-zero}, the minimum distance of $\C_{m}$ is $q-7$ and the code $\C_{m}$ has at most four nonzero weights, i.e.,
$$
w_1=q-7,~w_2=q-1,~w_3=q,~w_4=q+1~.
$$
We now determine the number $A_{w_i}$ of the codewords with weight $w_i$ in $\C_{m}$.
Since $\C_{m}^\perp$ has minimum distance $d=4$, the first four Pless Power Moments lead to the following system of equations:
\begin{eqnarray}\label{eq:moment}
\left\{
\begin{array}{l}
A_{w_1}+A_{w_2}+A_{w_3}+A_{w_4}=q^4 -1 \\
w_1 A_{w_1}+w_2 A_{w_2}+w_3 A_{w_3}+w_4 A_{w_4}=q^3(q-1)n \\
w_1 ^2 A_{w_1}+w_2 ^2 A_{w_2}+w_3 ^2 A_{w_3}+w_4 ^2 A_{w_4}=q^2 (q-1)n (qn-n+1)\\
w_1 ^3 A_{w_1}+w_2 ^3 A_{w_2}+w_3 ^3 A_{w_3}+w_4 ^3 A_{w_4}=q (q-1)n (q^2 n^2-2qn^2+3qn-q+n^2-3n+2),
\end{array}
\right.
\end{eqnarray}
where $n=q+1$. Solving the system of equations in (\ref{eq:moment}) yields
the weight enumerator in (\ref{cmw}).

(\uppercase\expandafter{\romannumeral2})
By the conclusions of (\uppercase\expandafter{\romannumeral1}) and Theorem \ref{thm-d4}, from Theorem \ref{thm-designGAMtheorem} we get that both $\C_{m}$ and $\C_{m}^\perp$ hold $3$-designs.
By (\ref{cmw}), the number of the codewords with weight $q-7$ in $\C_m$ is
$$A_{q-7}=\frac{1}{336} (q-1)^2 q (q+1).$$
Since $q-7$ is the minimum
weight of $\C_{m}$, the number of supports of the codewords of weight $q-7$ is
\begin{eqnarray}\label{b}
b=\frac{A_{q-7}}{q-1}=\frac{1}{336} (q-1) q (q+1).
\end{eqnarray}
Then the values of $\lambda$ follow from Equations (\ref{b}) and (\ref{eq:bb}).

By definitions, the complement of the supports of the codewords with the minimum weight $q-7$ in $\C_m$ holds a $3$-$(q+1,8,\lambda')$ and the number of this supports equals the value of $b$ of (\ref{b}). Then we deduce that $\lambda'=1$ from
$$
b=\frac{1}{336} (q-1) q (q+1)= \lambda' \cdot  \frac{\binom{q+1}{3}}{\binom{8}{3}}.
$$
This means that the complement of the supports of the minimum
weight codewords in $\C_m$ forms a Steiner system $S(3,8,q +1)$.
This completes the proof.

\end{proof}

\begin{example}\label{exam-2}
Let $m=2$. Then the code $\C_{2}$ has the parameters $[50, 4, 42]$ and weight
enumerator
$$1+16800 z^{42}+ 2469600 z^{48 } + 808800 z^{ 49}+ 2469600 z^{ 50},$$
which is verified by a Magma program.
\end{example}

\begin{example}\label{exam-3}
Let $m=3$. Then the code $\C_{3}$ has the parameters $[ 344, 4, 336 ]$ and weight
enumerator
$$1+41073858 z^{336}+ 5790693384 z^{342 } + 1971662832 z^{ 343}+ 6037857126 z^{344},$$
which is verified by a Magma program.
\end{example}

\subsection{ Automorphism groups of cyclic codes}

In this subsection, we will show that the cyclic code $\C_{m}$  and its dual $\C_{m}^\perp$  are invariant under group actions of certain permutation groups which are
$3$-transitive, i.e.,  the automorphism groups of those code are $3$-transitive. To this end, we use the similar method in Liu et al \cite{Liudingtang2021}.

Define
\begin{eqnarray}\label{eq-u(q+1)}
\mathrm{Stab}_{U_{q+1}}=
 \left\{ \left( \begin{array}{cc}
\beta_2 ^q & \beta_1 ^q\\
\beta_1 & \ \beta_2
\end{array} \right) \in \mathrm{PGL} (2,q^2): ~\beta_1, \beta_2 \in \gf(q^2),  \beta_1 ^{q+1} \neq \beta_2 ^{q+1}\right\}.
\end{eqnarray}
%According to the results in \cite{Liudingtang2021}, then
Then we have
\begin{eqnarray}\label{eq:stabilizers}
\mathrm{Stab}_{U_{q+1}}= \left( \begin{array}{cc}
u_0 & 1\\
1 & u_0
\end{array} \right) \PGL(2,  q)  \left( \begin{array}{cc}
u_0 & 1\\
1 & u_0
\end{array} \right)^{-1}
\end{eqnarray}
with $u_0\in U_{q+1}\setminus \{\pm 1\} $  and the following result which was documented  in \cite[Proposition 5]{Liudingtang2021}.

\begin{lemma}\cite{Liudingtang2021} \label{prop:Stab-U}
Let symbols and notation be the same as before. Let $q=7^m$  with $m\geq 2$ being a positive integer.  Then the setwise stabilizer of $U_{q+1}$ can be expressed as $\mathrm{Stab}_{U_{q+1}}$ defined by (\ref{eq-u(q+1)}). Moreover, the action of $\mathrm{Stab}_{U_{q+1}}$ on $U_{q+1}$ is equivalent to the action of $\PGL(2, q)$ on the projective line $\PG(1,  q)$ and $\mathrm{Stab}_{U_{q+1}}$ is $3$-transitive.
\end{lemma}

Denote the set
\begin{eqnarray}\label{eq:sf}
\begin{array}{c}
\mathcal{SF}:= \left\lbrace  \tr \left( a u^4+ b u^3\right) \in  \gf(q^2)[u]/\langle u^{q+1}-1 \rangle: a, b\in \gf(q^2)  \right\rbrace
\end{array}
\end{eqnarray}
and the operator $'\circ'$ is defined by
\begin{eqnarray}\label{eq:action-circ}
\begin{array}{c}
(G\circ f) (u):= (\beta_1 u+\beta_2)^{4(q+1)} f\left(\frac{\alpha_1 u+\alpha_2}{\beta_1 u+\beta_2}\right),
\end{array}
\end{eqnarray}
where
$G=\left( \begin{array}{cc}
\alpha_1 & \alpha_2\\
\beta_1 & \ \beta_2
\end{array} \right)^{-1} \in \mathrm{GL} (2,q^2)$
and $f\in \mathcal{PF}$.

Let $\bar{G}=\left( \begin{array}{cc}
\beta_2 ^q & \beta_1 ^q\\
\beta_1 & \ \beta_2
\end{array} \right)^{-1} \in \mathrm{GL} (2,q^2)$
and
denote
\begin{eqnarray}
\mathrm{\overline{Stab}}_{U_{q+1}}=
 \left\{ \left( \begin{array}{cc}
\beta_2 ^q & \beta_1 ^q\\
\beta_1 & \ \beta_2
\end{array} \right) \in \mathrm{GL} (2,q^2) : ~\beta_1, \beta_2 \in \gf(q^2),  \beta_1 ^{q+1} \neq \beta_2 ^{q+1}\right\}.
\end{eqnarray}

Next we will show that the linear space $\mathcal{PS}$ under the action of the group $\mathrm{\overline{Stab}}_{U_{q+1}}$  is invariant. To this end, we need the results in the following two lemmas.

%the operator the action $'\circ'$ defined in (\ref{eq:action-circ}) is a representation of $ \mathrm{\overline{Stab}}_{U_{q+1}}$ on the linear space $\mathcal{PS}$. To this end, we need the results in the following two lemmas.

\begin{lemma}\label{lem:action-linear}
Let $q=7^m$ with $m\geq 2$ being a positive integer.
Let $\bar{G}=\left( \begin{array}{cc}
\beta_2 ^q & \beta_1 ^q\\
\beta_1 & \ \beta_2
\end{array} \right)^{-1} \in \mathrm{GL} (2,q^2)$  and $f\in \mathcal{SF}$.  Then $\bar{G}\circ f \in \mathcal{SF}$.
\end{lemma}
\begin{proof}
Denote $f_1=\tr(au^4) \in  \gf(q^2)[u]/\langle u^{q+1}-1\rangle$  and $f_2=\tr(bu^3) \in  \gf(q^2)[u]/\langle u^{q+1}-1\rangle$. We only need to prove
$(\bar{G}\circ f_1)(u)\in \mathcal{SF}$ and $(\bar{G}\circ f_2)(u)\in \mathcal{SF}$ for any $(a,b)\in \gf(q^2)^2$.

For any $a \in \gf(q^2)$, from definitions we have
\begin{eqnarray}\label{Eqn:A-g}
\begin{array}{l}
(\bar{G}\circ f_1)(u) \\
= (\beta_1 u+\beta_2)^{4(q+1)} f_1\left(\frac{\beta_2 ^q  u+\beta_1 ^q}{\beta_1 u+\beta_2}\right),\\
=\tr\left( a \cdot (\beta_1 u+\beta_2)^{4q}  (\beta_1 u+\beta_2)^{4}  \cdot  \frac{(\beta_2 ^q  u+\beta_1 ^q)^4}{(\beta_1 u+\beta_2)^4}                       \right ) \\
=\tr\left( a \cdot (\beta_1^q u^{-1}+\beta_2^q)^{4}  \cdot  u^4 (\beta_2 ^q +\beta_1 ^q u^{-1})^4  \right)  \\
=\tr\left( a u^4 \cdot (\beta_1^q u^{-1}+\beta_2^q)^{8}   \right)
\end{array}
\end{eqnarray}
Note that $7| \binom{8}{i}$ for any $i \in \{2,3,4,5,6\}$.
Therefore, from the Binomial Theorem we have
\begin{eqnarray}\label{Eqn:A-f1}
u^4 \cdot (\beta_1^q u^{-1}+\beta_2^q)^{8}&=& \sum_{i=0}^{8} \binom{8}{i} \beta_1 ^{qi} \beta_2^{q(8-i)} u^{4-i},  \nonumber \\
&=& \beta_2^{8q} u^4+\beta_2^{7q} \beta_1^{q} u^3+ \beta_1^{8q} u^{-4}+\beta_2^{q} \beta_1^{7q} u^{-3} \nonumber \\
&=& \beta_2^{8q} u^4+\beta_2^{7q} \beta_1^{q} u^3+ (\beta_1^{8} u^{4})^q+(\beta_2 \beta_1^{7} u^{3})^q
\end{eqnarray}
Applying the above equation (\ref{Eqn:A-f1}) to (\ref{Eqn:A-g}), we get
\begin{eqnarray}\label{Eqn:A-f11}
\begin{array}{l}
(\bar{G}\circ f_1)(u)
=\tr\left( a(\beta_2^{8q} u^4+\beta_2^{7q} \beta_1^{q} u^3)+ a^{1/q}(\beta_1^{8} u^{4}+ \beta_2 \beta_1^{7} u^{3})   \right)  \in \mathcal{SF}.
\end{array}
\end{eqnarray}
Using the similar method on $f_1$ to $f_2$, we can easily obtain

\begin{eqnarray}\label{Eqn:A-f21}
\begin{array}{l}
(\bar{G}\circ f_2)(u)
=\tr\left( b(\beta_1 \beta_2^{7q} u^4+\beta_2^{7q+1} u^3)+ b^{1/q}(\beta_1^{7}\beta_2^{1/q} u^{4}+ \beta_1^{q+7} u^{3})   \right)  \in \mathcal{SF}
\end{array}
\end{eqnarray}
for any $b \in \gf(q^2)$. The desired conclusion then follows from Equations (\ref{Eqn:A-f11}) and (\ref{Eqn:A-f21}).
\end{proof}

According to Lemma \ref{lem:action-linear} and the definition of  $'\circ'$ in (\ref{eq:action-circ}), we can easily obtain the following result and we omit its proof.

\begin{lemma}\label{lem:representation}
Let symbols and notation be the same as before. Let $q=7^m$  with $m\geq 2$ being a positive integer and $E$ be the $2\times 2$ identity matrix. For any $\overline{G}_1, \overline{G}_2 \in  \mathrm{\overline{Stab}}_{U_{q+1}}$ and $f_1,  f_2 \in \mathcal{SF}$, we have $\overline{G}_1 \circ f_1 \in \mathcal{SF}$ , $E \circ f_1= f_1$ , $(\overline{G}_1 \overline{G}_2) \circ f_1 = \overline{G}_1 \circ  (\overline{G}_2 \circ  f_1 )$,
and $\overline{G}_1 \circ  ( a' f_1+b' f_2)=a' \overline{G}_1 \circ  f_1 +b' \overline{G}_2 \circ  f_2$ for all $a',b' \in \gf(q)$.
\end{lemma}
%\begin{proof}
%The desired conclusions follow form Lemma \ref{lem:action-linear} and the definition of  $'\circ'$ in (\ref{eq:action-circ}).
%\end{proof}

Let $\gf(q)^{U_{q+1}}$ denote the vector space consisting  of all elements $(c_u)_{u \in U_{q+1}}$,
where $c_u\in \gf(q)$.  The action of the  semidirect product $\left( \gf(q)^* \right)^{U_{q+1}} \rtimes \mathrm{Stab}_{U_{q+1}}$
on $\gf(q)^{U_{q+1}}$
is defined by
\begin{eqnarray*}
\left((a_u)_{u\in U_{q+1}};  g\right) (c_u)_{u\in U_{q+1}}  = (a_{u} c_{g^{-1}(u)})_{u \in U_{q+1}}.
\end{eqnarray*}
Then the multiplication in $\left( \gf(q)^* \right)^{U_{q+1}} \rtimes \mathrm{\overline{Stab}}_{U_{q+1}}$ is given by
\begin{eqnarray*}
\left((a_u)_{u\in U_{q+1}};  g_1\right) \left((b_u)_{u\in U_{q+1}};  g_2\right) = \left((c_u)_{u\in U_{q+1}};  g_1g_2\right),
\end{eqnarray*}
where $c_u=a_u b_{g_1^{-1}(u)}$.

The following theorem is one of the main result in this paper. It show that the code $\C_m$ and its dual admit $3$-transitive automorphism group.

\begin{theorem}\label{thm:C-group}
Let $q=7^m$ with $m\geq 2$ being a positive integer.
Define the subgroup $G_{i}$ of $\left( \gf(q)^* \right)^{U_{q+1}} \rtimes \mathrm{Stab}_{U_{q+1}}$ by
\begin{eqnarray*}
G_{i}=\left\{ \left(\left((\beta_1u+\beta_2)^{4i(q+1)}\right)_{u\in U_{q+1}}; \left( \begin{array}{cc}
\beta_2^q & \beta_1^q\\
\beta_1 & \beta_2
\end{array} \right) ^{-1} \right): \beta_1,  \beta_2 \in \gf(q^2),  \beta_1^{q+1} \neq  \beta_2^{q+1}\right\},
\end{eqnarray*}
where $i\in \{1,-1\}$. Then we have the following results.
\begin{enumerate}
  \item $G_{1}$ is a subgroup of the monomial automorphism group $\mathrm{MAut}(\C_{m})$. Moreover, the automorphism group of $\C_{m}$ is $3$-transitive.
  \item $G_{-1}$ is a subgroup of the monomial automorphism group $\mathrm{MAut}(\C_{m}^{\perp})$ and the automorphism group of $\C_{m}$ is $3$-transitive.
\end{enumerate}
\end{theorem}

\begin{proof}
1) By the definitions in (\ref{cm}) and (\ref{eq:sf}),
$$\C_m=\left\{ (f(u))_{u \in U_{q+1}}:  f \in  \mathcal{SF} \right\}.$$
Then the desired conclusion follows from definitions and Lemmas \ref{lem:representation}, \ref{eq-maut} and \ref{prop:Stab-U}.

2) The desired conclusion follows from the first conclusion of this theorem.

\end{proof}

%It comes  straightforwardly from Lemma \ref{lem:action-linear} and the definition of the action $'\circ'$.

\section{Concluding remarks}\label{sec-summary}

In this paper, we investigated a class of cyclic codes $\C_m$ over $\gf(7^m)$ and completely determined their parameters. The results showed that the code $\C_m$ has four nonzero weights and supports $3$-designs. Meanwhile, the dual code of $\C_m$  also supports $3$-designs, and the automorphism group of the code $\C_m$ and its dual $\C_m ^\perp$  are $3$-transitive.
Specifically, the complements of the supports of the minimum weight codewords in $\C_m$ form a Steiner system $S(3,8, 7^m+1)$. Using the similar method of this paper and \cite{Liudingtang2021},
we remark that it may obtain some new cyclic codes admitting $3$-transitive automorphism groups  and determine their parameters by properly choosing the value of $q$ and the exponents of $u$ in (\ref{cm}).


\begin{thebibliography}{99}


\bibitem{Chien5}
R. T. Chien, ``Cyclic decoding procedure for the Bose-Chaudhuri-Hocquenghem codes'', \emph{IEEE Trans. Inform. Theory}, vol. 10, no. 4, pp. 357--363, 1964.

\bibitem{Forney12}
G. D. Forney, ``On decoding BCH codes,'' \emph{IEEE Trans. Inform. Theory,} vol. 11, no. 4, pp. 549--557, 1995.

\bibitem{Prange28}
E. Prange, ``Some cyclic error-correcting codes with simple decoding algorithms,'' Air Force Cambridge Research Center-TN-58-156, Cambridge, Mass., April 1958.

\bibitem{dinghell2013}
C. Ding, T. Helleseth, ``Optimal Ternary Cyclic Codes From Monomials,'' \emph{IEEE Trans. Inform. Theory},
vol. 59, no. 9, pp. 5898-5904, 2013.
\bibitem{zhouding2013}
Z. Zhou, C. Ding, ``Seven Classes of Three-Weight Cyclic Codes,'' \emph{IEEE Trans. Commun.}, vol. 61, no. 10, pp. 4120-4126, 2013.
\bibitem{Liding2013}
C. Li, C. Ding, S. Li, ``LCD Cyclic Codes Over Finite Fields,'' \emph{IEEE Trans. Inform. Theory}, vol. 63, no. 7, pp. 4344-4356, 2017.
\bibitem{LiuDing2017}
H. Liu, C. Ding, ``Infinite families of 2-designs from GA1(q) actions,'' ArXiv:1707.02003, 2017.
\bibitem{Ding2018}
C. Ding, ``A sequence construction of cyclic codes over finite fields,'' \emph{Cryptogr. Commun.}, vol. 10, no. 2, pp. 319-341, 2018.
\bibitem{Zha2021}
Z. Zha, L. Hu, Y. Liu, X. Cao,  ``Further results on optimal ternary cyclic codes,'' \emph{Finite Fields Their Appl.} , vol.75, pp. 101898, 2021.


\bibitem{Dingbook18}
C. Ding, \emph{Designs from Linear Codes}. Singapore: World Scientific, 2018.

\bibitem{Dingtang2019}
C. Ding, C. Tang, ``Infinite families of near MDS codes holding t-designs,'' \emph{IEEE Trans. Inform. Theory}, vol. 66, no. 9, pp. 5419--5428, 2020.

\bibitem{Dingtv2020}
C.Ding, C. Tang, V.D. Tonchev, ``Linear codes of 2-designs associated with subcodes of the ternary generalized Reed-Muller codes,'' \emph{Des. Codes Cryptogr.}, vol. 88, no. 4, pp. 625-641, 2020.
\bibitem{Dingt20201}
C. Ding, C. Tang, ``The linear codes of t-designs held in the Reed-Muller and Simplex codes,'' arXiv:2008.09935, 2020.



\bibitem{Ding18dcc}
C. Ding, ``Infinite families of 3-designs from a type of five-weight code," \emph{Des. Codes Cryptogr.}, vol. 86, no. 3, pp. 703--719, 2018.

\bibitem{DingLi16}
C. Ding, C. Li, ``Infinite families of $2$-designs and $3$-designs from linear codes,''
\emph{Discrete Math.}, vol. 340, no. 10, pp. 2415--2431, 2017.

\bibitem{Ding15}
C. Ding,  \emph{Codes from Difference Sets}. Singapore: World Scientific, 2015.
\bibitem{ton1}
V. D. Tonchev, “Codes and designs,” In Handbook of Coding Theory, vol. II, V. S. Pless and W. C. Huffman, eds., Elsevier, Amsterdam, pp. 1229–1268, 1998.
\bibitem{ton2}
V. D. Tonchev, “Codes,” In Handbook of Combinatorial Designs, 2nd edition, C. J. Colbourn and J. H. Dinitz, eds., CRC Press, New York, pp. 677–701, 2007.


\bibitem{Ding18jcd}
C. Ding, ``An infinite family of Steiner systems from cyclic codes," \emph{Journal of Combinatorial Designs}, vol. 26, pp. 127--144, 2018.



\bibitem{HP10}
W. C. Huffman, V. Pless,  \emph{Fundamentals of Error-Correcting Codes}. Cambridge: Cambridge University Press, 2003.

\bibitem{TDX2019}
C. Tang, C. Ding, M. Xiong, ``Steiner systems $S(2, 4, \frac{3^m-1}{2})$ and 2-designs from ternary linear codes of length $\frac{3^m-1}{2}$,'' \emph{Des. Codes Cryptogr.}, vol. 87, no. 12, pp. 2793--2811, 2019.

\bibitem{Tangit2019}
C. Tang, C. Ding, M. Xiong, ``Codes, differentially $\delta$-uniform functions, and $t$-designs," \emph{IEEE Trans. Inform. Theory},  vol. 66, no. 6, pp. 3691--3703, 2020.

\bibitem{Tangding2020}
 C. Tang, C. Ding, ``An infinite family of linear codes supporting $4$-designs,'' \emph{IEEE Trans. Inform. Theory}, vol. 67, no. 1, pp. 244-254, 2021.

\bibitem{du1}
 X. Du, R. Wang, C. Tang, Q. Wang,
``Infinite families of 2-designs from two classes of binary cyclic codes with three nonzeros'', arXiv:1903.08153, 2019.

\bibitem{du2}
 X. Du, R. Wang, C. Tang, Q. Wang, ``Infinite families of 2-designs from two classes of linear codes'', arXiv:1903.07459, 2019.

\bibitem{Yan01}
H. Yan,  ``A class of primitive BCH codes and their weight distribution.'' \emph{ Appl. Algebra Eng. Commun. Comput.} vol. 29, no. 1, pp. 1--11, 2018.

\bibitem{Yan02}
H. Yan, H. Liu, C. Li, et al., ``Parameters of LCD BCH codes with two lengths'', \emph{Adv. Math. Commun.}, vol. 12, no. 3, pp. 579-594, 2018.


\bibitem{Liudingtang2021}
Q. Liu, C. Ding, S. Mesnager, C. Tang, V.D. Tonchev, ``On infinite families of narrow-sense antiprimitive BCH codes admitting 3-transitive automorphism groups and their consequences'', arXiv:2109.09051, 2021.
\end{thebibliography}
\end{document}